\documentclass[a4paper,11pt,american]{article} 
\usepackage{anysize}

\usepackage[utf8]{inputenc} 
\usepackage{utf8math}

\usepackage{inconsolata}
\usepackage{libertine}
\usepackage{mathrsfs}
\usepackage[american]{babel}
\usepackage{mathtools} 
\usepackage{amssymb}
\usepackage{hyperref}
\usepackage{amscd}
\usepackage{todonotes}

\usepackage{amsthm}

\usepackage{enumerate}

\usepackage{framed}

\newcommand{\R}{\mathbb{R}}
\newcommand{\Z}{\mathbb{Z}}

\DeclareMathOperator{\dist}{dist}

\DeclareMathOperator{\poly}{poly}

\theoremstyle{plain}
\newtheorem{theorem}{Theorem}

\theoremstyle{definition}

\newtheorem{remark}[theorem]{Remark}

\title{An improved bound on sums of square roots via the subspace theorem }

\author{
Friedrich Eisenbrand\thanks{EPFL, Switzerland, \texttt{friedrich.eisenbrand@epfl.ch}}
\and
Matthieu Haeberle\thanks{EPFL, Switzerland, \texttt{matthieu.haeberle@epfl.ch}}
\and
Neta Singer\thanks{EPFL, Switzerland, \texttt{neta.singer@epfl.ch}}
}

\begin{document}
\maketitle

\begin{abstract}
\noindent 
The \emph{sum of square roots} is  as follows:  Given
 $x_1,\dots,x_n ∈ ℤ$ and  $a_1,\dots,a_n ∈ ℕ$
 decide whether $ E= ∑_{i=1}^n x_i \sqrt{a_i} ≥0$. 
  It is a prominent open problem (Problem 33 of the
\emph{Open Problems Project}), whether this can be decided in
polynomial time.  The state-of-the-art methods rely on separation
bounds, which are lower bounds on the minimum nonzero absolute value
of $E$.  The current best bound shows that $ |E| ≥ \left(n ⋅ \max_i (|x_i| ⋅\sqrt{a_i}) \right)^{-2^n} $, which is doubly exponentially
small. 

We provide a new bound of the form
$|E| ≥ γ ⋅ (n ⋅ \max_i |x_i|)^{-2n}$ where $γ$ is a constant
depending on $a_1,\dots,a_n$. This is singly exponential in $n$ 
for fixed $a_1,\dots,a_n$. The constant $\gamma$ is not explicit and stems
from the \emph{subspace theorem}, a deep result in the \emph{geometry
  of numbers}.
\end{abstract}

\section{Introduction}
\label{sec:introduction}

Several geometric optimization problems such as \emph{euclidean
  traveling salesman} or \emph{euclidean shortest path}, rely on
comparisons of \emph{sums of square roots}, which is a decision problem as follows.  Given integers $x_1,\dots,x_n ∈ ℤ$ and positive integers $a_1,\dots,a_n ∈ ℕ$ decide whether 
\begin{equation}
  \label{eq:1}
    E = ∑_{i=1}^n x_i \sqrt{a_i} ≥ 0. 
\end{equation}
%
While the decision problem~\eqref{eq:1} is easy to state, it is not known to be decidable in polynomial time on a
Turing machine, nor is it known to be NP~\cite{garey1976some}, see also~\cite{o1981advanced,etessami2009recursive,allender2009complexity}.  The best known complexity class containing the decision problem~\eqref{eq:1} is PSPACE. This follows by modeling the decision as a problem in the \emph{existential theory of the reals} for which a PSPACE-algorithm exists~\cite{canny1988some,renegar1988faster}. The \emph{zero test} (when $≥0$ is replaced by $=0$) can be decided in polynomial time with an algorithm of Blömer~\cite{blomer1991computing,blomer1998probabilistic}.

The state-of-the-art method to decide~\eqref{eq:1} is based on \emph{separation bounds}, see, e.g.~\cite{li2005recent}. Separation bounds are lower bounds on the absolute value of $E$, defined in~\eqref{eq:1}, when it is nonzero. The best known bound in our setting is by Burnikel et al.~\cite{burnikel2000strong}. The bound follows from the fact that the product of the \emph{conjugates}, see, e.g.~\cite{lang2012algebra}, of $E$ is an integer. Each conjugate of $E$ is of the form  
\begin{displaymath}
   ∑_{i=1}^n y_i⋅ x_i \sqrt{a_i}, \quad y ∈ \{\pm 1\}^n,
 \end{displaymath}
 of absolute value bounded by $n \max\limits_i (|x_i| \sqrt{a_i})$. This implies that 
 \begin{equation}
   \label{eq:3}
   |E| ≥ \left(n \max_i \left( |x_i| \sqrt{a_i}\right) \right)^{-(2^n-1)}   
 \end{equation}
 whenever $E ≠0$. If follows that~\eqref{eq:1} can be decided with an
 approximation of the numbers $x_i \sqrt{a_i}$ in which
 $\mathcal{O}\left[ 2^n \log \left(n \max\limits_i \left( |x_i| \sqrt{a_i}\right)\right)\right]$
 bits of the respective fractional parts are correct. This bound is
 exponential in $n$.

 On the other hand, there is no empirical
 evidence~\cite{cheng2010bounding} that the reciprocal $1/|E|$ can be
 doubly exponential in $n$. The best empirical lower bounds~\cite{qian2006much} 
 observed for  $1/|E|$ are of the form $ \left(\max_i a_i\right)^{Ω(n)}$. The question of whether singly-exponential separation bounds for $|E|$
 exist, is a highly visible open problem in
 computing~\cite{o1981advanced}, see also~\cite[Problem 33]{OpenProblems}.

 \subsection*{Contribution of this paper}
 
Our main result is a new separation bound for $|E|$ that shows single-exponential dependence on $n$ if $a_1,\dots,a_n$ are fixed. More precisely, we show the following. 

 \begin{enumerate}[i)] 
 \item If $E≠0$, then
   \begin{equation*}
     \label{eq:4}
     |E| ≥ \left( \frac{1}{n ⋅ \|x\|_∞}\right)^{2n} ⋅ γ, \, \begin{aligned} & \textrm{   where $γ ∈ ℝ_+$ is a constant depending on $\sqrt{a_1},\dots,\sqrt{a_n}$}.%
     \end{aligned}  
   \end{equation*}
  
 \end{enumerate}
 Compared to the bound~\eqref{eq:3} of Burnikel et al. this decreases the dependence on $\|x\|_∞$ from doubly exponential in $n$ to exponential. The new bound is obtained by applying tools and concepts from the \emph{geometry of numbers} such as \emph{lattices}, \emph{Minkowski's first and second theorem}, and  Schmidt's~\cite{schmidt1972norm}  celebrated \emph{subspace theorem}. This bound is asymptotically (almost)  tight in the following sense. 
 
 \begin{enumerate}[i)]
   \setcounter{enumi}{1}
 \item  For each $L ∈ ℕ_{\geq 2}$ there exists $x ∈ℤ^n$,  $x ≠0$  with $\|x\|_∞ ≤ L$ with 
   \begin{equation*}
     0 < |E| ≤ \frac{n \max_i\sqrt{a_i} }{L^{n-1}}.   %
     %
   \end{equation*}
 \end{enumerate}
This  bound follows form the pigeon-hole principle, similar to its application to the \emph{number-balancing problem}~\cite{karmarkar1982differencing,hoberg2017number}.

\begin{remark}
   The format of Problem~33 in~\cite{OpenProblems} differs slightly
   from the problem description~\eqref{eq:1}  in this
   paper. Using our notation, Problem~33 requires $n$ to be even and
   $x_i ∈ \{\pm 1\}$ for $i=1,\dots,n$. Furthermore exactly half of
   the $x_i$ are positive one. However, the question whether the logarithm of
   the reciprocal of the best separation bound is exponential or
   sub-exponential in $n$ is equivalent in both settings. The
   problem~\eqref{eq:1} can be reformulated in the format of
   Problem~33 by replacing $x_i \sqrt{a_i}$ with $x_i ≠0$  with
   $(x_i / |x_i|)  \sqrt{x_i^2 a_i}$ and doubling $n$ if necessary.  
 \end{remark}

 \subsubsection*{Simplifying assumptions}
 \label{sec:simpl-assumpt}
 
Before we develop the connection of separation bounds for~\eqref{eq:1} to the geometry of numbers, we justify simplifying assumptions on the input of~\eqref{eq:1}.   
If some $a_i$ is divisible by a square $y^2$ with $y ∈ ℕ \setminus \{0, 1\}$, then $a_i$ can be replaced by $a_i / y^2$ as long as $x_i$ is replaced by $x_i ⋅ y$. Furthermore, if $a_i = a_j$ for $i ≠j$, then we can delete $a_j$ and replace $x_i$ with $x_i + x_j$, thereby reducing the dimension $n$ that appears in the exponent of our bound.  We can therefore assume, without loss of generality, that each $a_i∈ ℕ$ is  \emph{square-free} and that the $a_i$ are \emph{distinct}.
We recall the following fact, see~~\cite{besicovitch1940linear}. 
\begin{theorem}
  Let $a_1,\dots,a_n ∈ℕ_+$ be distinct square-free integers. The set 
  \begin{displaymath}
    \left\{\sqrt{a_1},\dots,\sqrt{a_n} \right\}
  \end{displaymath}
is linearly independent over the rational numbers $ℚ$. 
\end{theorem}

 \section{Lattices and separation bounds}
 \label{sec:latt-theor-mink}
 
 Let $A ∈ ℝ^{n × n}$ be a matrix of full rank. The set $Λ(A) = \{A x ： x∈ ℤ^n\}$ is the \emph{lattice} generated by the \emph{(lattice) basis} $A$. If $A∈ ℚ^{n ×n}$ is rational, then the lattice $Λ(A)$ is rational.  A \emph{shortest vector} w.r.t. a norm $\| ⋅ \|$ of a lattice $Λ ⊆ ℝ^n$ is a nonzero $v ∈ Λ$ of minimal norm. 
 Lattices have been used in the context of computing separation bounds by  Cheng et al.~\cite{cheng2010bounding}. Here, the main idea is to consider the lattice generated by the basis
 \begin{equation}
   \label{eq:5a}
   \begin{pmatrix}
     N & -N \sqrt{a_1} & \cdots &  -N \sqrt{a_n} \\
     &  1            &        &   \\
     &               & \ddots & \\
     &               &        & 1
   \end{pmatrix}
 \end{equation}
 where $N ∈ ℕ_+$ is a positive integer. Suppose one is interested in the minimum absolute value of $E$ in~\eqref{eq:1} where the $x_i$ are bounded by one in absolute value. If the length of the shortest vector w.r.t. $ℓ_2$ is larger than $\sqrt{n+1}$, then $1/N$ is a lower bound on $E$ in that case. Using algorithms for computing or approximating shortest vectors in the $ℓ_2$-norm~\cite{lenstra1982factoring,schnorr1987hierarchy} can then be used to find the smallest such $N$. The approach of Cheng et al.~\cite{cheng2010bounding} is suitable  for computing good lower bounds for large instances of the sum-of-square-roots problem.

 Our approach is based on the dual  of the lattice generated by~\eqref{eq:5a}. Recall that the \emph{dual}  lattice of $Λ ⊆ ℝ^n$ is the lattice
 \begin{displaymath}
   Λ^*  = \{ y ∈ ℝ^n ： y^T v ∈ℤ \text { for each } v ∈Λ\}. 
 \end{displaymath}
 If $Λ$ is generated by $A ∈ℝ^{n × n}$, then $Λ^*$ is generated by $A^{-T}$, see, e.g.~\cite{cassels2012introduction}. Let $Q = N^{1/(n+1)}$ and denote $β_i = \sqrt{a_i}$ for $i=1,\dots,n$. The dual of the lattice generated by \eqref{eq:5a} is thus generated by the basis 
 \begin{equation}
   \label{eq:5b} B = 
   \begin{pmatrix}
     1/Q^{n+1} &  &  &   \\
     β_1      &  1            &        &   \\
     \vdots     &               & \ddots & \\
     β_n    &               &        & 1
   \end{pmatrix}
 \end{equation}
 Let $\| ⋅\|$ be a norm and $i ∈ \{1,\dots,n\}$. The \emph{$i$-th successive minimum} of $Λ$  is the smallest radius $R>0$ such that $\{ x ∈ ℝ^n ： \|x\|≤R\}$ contains $i$ linearly independent lattice vectors.  $i$-th successive minimum is denoted by $λ_i$. In the following, we will restrict our attention to the successive minima w.r.t. the $ℓ_∞$-norm.

The absolute value of the determinant of any basis of a lattice $Λ ⊆ ℝ^n$ is an invariant of the lattice. It is called the \emph{lattice determinant} and is denoted by $\det(Λ)$. The following is referred to as Minkowski's second theorem. We state it for the $ℓ_∞$-norm, see~\cite{cassels2012introduction}. 

\begin{theorem}[Minkowski's theorem for $ℓ_∞$]
  \label{thr:2}
  Let $Λ ⊆ ℝ^n$ be a lattice. One has
  \begin{equation}
    \label{eq:7}
    λ_1 \cdots λ_n ≤ \det(Λ), 
  \end{equation}
  where the successive minima $λ_i$ are with respect to the $ℓ_∞$-norm. 
  In particular, one has
  \begin{displaymath}
    λ_1 ≤ \det(Λ)^{1/n}. 
  \end{displaymath}
\end{theorem}

We now develop the connection between separation bounds for~\eqref{eq:1} and the theory described so far. Observe that the determinant of the lattice generated by the basis $B$ in~\eqref{eq:5b} is $1/Q^{n+1}$ and that the dimension is $n+1$. Theorem~\ref{thr:2} implies that $Λ(B)$ contains a nonzero lattice vector $v$ with $\|v\|_∞ ≤ 1/Q$. If this bound is almost tight, then the value of $Q$ carries over to a separation bound for $E$ in~\eqref{eq:1}. This is our next theorem. 

\begin{theorem}
  \label{thr:1}
Consider $ E = ∑_{i=1}^n x_i \sqrt{a_i}$ with  
$a_1,\dots,a_n∈ ℕ$  and $x = (x_1,\dots,x_n)~∈~ℤ^n~\setminus~\{0\}$ and let $Λ(B)$ be the lattice generated by $B$ in~\eqref{eq:5b}. 

   \medskip
   \noindent 
  If $Q ≥ \left(2 n \|x\|_∞\right)^{3/2}$ and if $λ_1 ≥ 1/Q^{1 + \frac{1}{3n}}$, then
  \begin{equation}
    |E| ≥ \frac{1}{Q^{n+1}}. 
  \end{equation}
\end{theorem}

\begin{proof}
  Minkowski's second theorem gives the bound 
  \begin{align}\label{eq:minkowski2}
    \prod_{i=1}^{n+1} \lambda_i \leq \frac{1}{Q^{n+1}}.
  \end{align}
 Since   $\lambda_1 \geq 1/Q^{1+\frac{1}{3n}}$ one has 
\[ \lambda_{i}  \leq \frac{1}{Q^{2/3}} \text{  for each  } i ∈\{1,\dots,n+1\}.\]
The successive minima are attained at  $n+1$ linearly independent lattice vectors. Therefore,  
one of the successive minima is attained at a lattice vector
\begin{displaymath}
  v =
B \cdot \begin{pmatrix} q \\ -p \end{pmatrix}
\end{displaymath}
with $q ∈ ℕ$ and $p = (p_1,\dots,p_n)^T ∈ ℤ^n$ such that $p^T x ≠0$. Since $p^Tx \in \Z$, one has
\begin{equation}
  \label{eq:10}
  |p^Tx| ≥1.
\end{equation}
The condition $\|v\|_∞ ≤ 1/Q^{2/3}$ implies that
\begin{eqnarray*}
  |q ⋅ β_i - p_i| ≤   \frac{1}{Q^{2/3}} \leq \frac{1}{2 n \|x\|_∞}  \text{  for each  } i ∈\{1,\dots,n+1\}.
\end{eqnarray*}
By the triangle inequality, 
\begin{displaymath}
 |q β^T x - p^T x | ≤ \frac12
\end{displaymath}
which, together with $|p^T x|≥1$ implies that
\begin{equation}
  \label{eq:8}
  |β^T x | ≥ \frac{1}{2 q}.
\end{equation}
On the other hand, $\|v\|_∞ ≤ 1/ Q^{2/3}$ implies that $q ≤ Q^{n+\frac{1}{3}}$. The claim follows with~\eqref{eq:8} and since $2 ⋅Q^{n+\frac{1}{3}} ≤ Q^{n+1}$ for $Q ≥ \left(2 n \|x\|_∞\right)^{3/2}≥ 2^{3/2}$.
\end{proof}

\begin{remark}
  \label{rem:1}
  This proof generalizes  the main idea of a technique of Frank and Tardos~\cite{frank1987application}. An integer vector $p ∈ ℤ^n$ stemming from   $(q,p^T) ∈ ℕ × ℤ^n$  that is a sufficiently good simultaneous approximation to a real vector $β ∈ ℝ^n$,  separates the same set of integer points $y$ with bounded infinity norm,  as long as $p^T y ≠0$. We use this principle in (\ref{eq:10}), when all successive minima are sufficiently good approximations. 
\end{remark}

\subsection*{Using the subspace theorem}
\label{sec:subspace}

We consider again a lattice vector in $ v ∈ Λ(B)$ 
\begin{equation}
  \label{eq:9}
   v = \begin{pmatrix}
     1/Q^{n+1} &  &  &   \\
     β_1      &  1            &        &   \\
     \vdots     &               & \ddots & \\
     β_n    &               &        & 1
   \end{pmatrix} ⋅
  \begin{pmatrix}
    q\\
    -p_1\\
    \vdots\\
    -p_n\\
  \end{pmatrix}
\end{equation}
with $q ∈ ℕ$ and $p = (p_1,\dots,p_n)^T ∈ ℤ^n$. Minkowski's bound  $λ_1≤ 1/Q$ implies the theorem of Dirichlet on \emph{simultaneous Diophantine approximation}, see~\cite{schmidt1996diophantine}. 
\begin{theorem}[Dirichlet's Theorem] 
Given $β_1,\dots,β_n ∈ℝ$ and $Q ∈ ℕ_+$, there exist integers $q,p_1,\dots,p_n ∈ℤ$ with 
\end{theorem}
\begin{enumerate}[i)]
\item $1 ≤ q ≤ Q^n$ and
\item $| q β_i -p_i| ≤ 1/Q$ for $i=1,\dots,n$.    
\end{enumerate}
\medskip 
The \emph{subspace theorem} of Wolfgang Schmidt~\cite{schmidt1972norm} implies a lower bound that is almost tight. 
\begin{theorem}[Theorem 1B in~\cite{schmidt1996diophantine}]
  \label{thr:3}
  Let $β_1,\dots,β_n ∈ ℝ$ be real algebraic numbers such that $\{1,β_1,\dots,β_n\}$ is linearly independent over $ℚ$ and let $δ>0$. There are only finitely many positive integers $q ∈ ℕ_+$ such that
  \begin{displaymath}
    q^{1+δ} \dist_ℤ(q⋅β_1) \cdots \dist_ℤ(q⋅β_n) <1. 
  \end{displaymath}
\end{theorem}
Here $\dist_ℤ(x)$ is the  distance of the real number $x ∈ ℝ$ to the integers. 
It remains to show that there exists a good $Q$ satisfying the conditions of Theorem~\ref{thr:1}, which together with Theorem \ref{thr:3} will prove our main result.
\begin{theorem}
  \label{thr:4}
  Consider $ E = ∑_{i=1}^n x_i \sqrt{a_i}$ with  
  $a_1,\dots,a_n∈ ℕ$  and  $x = (x_1,\dots,x_n)~∈~ℤ^n~\setminus~\{0\}$. There exists a constant $γ ∈ ℝ$ depending on $a_1,\dots,a_n$ such that $E ≠0$ implies
  \begin{equation}
    \label{eq:11}
    |E| ≥  \left( \frac{1}{n ⋅ \|x\|_∞}\right)^{2n} ⋅ γ. 
  \end{equation}
\end{theorem}

\begin{proof}
  Following the arguments in Section~\ref{sec:simpl-assumpt} we can assume that the $a_i$ are distinct square-free integers. And assume for now that all $a_i$ are different from one. This implies that the set
  \begin{displaymath}
    \left\{1, β_1 = \sqrt{a_1} ,\dots, β_n = \sqrt{a_n}\right\} 
  \end{displaymath}
  is linearly independent over $ℚ$.
  It remains to show that there exists some $Q_0 \in ℕ_+$ such that the first successive minimum $λ_1$ of the lattice $Λ(B)$ verifies $λ_1 ≥ 1/Q^{1 + \frac{1}{3n}}$ for all $Q \geq Q_0$. 
The assertion then follows with Theorem~\ref{thr:1} applied to $Q = (Q_0 \cdot \left(2n\| x\|_\infty\right)^{3/2})$. To this end, 
  let $\delta= \frac{1}{3n}$ and suppose to the contrary that the first successive minimum of $Λ(B)$  satisfies
  \begin{displaymath}
    λ_1 ≤ \frac{1}{Q^{1+δ}}.
  \end{displaymath}
  This means that there exists a $q ∈ ℕ_+$ with 
  \begin{enumerate}[i)] 
  \item $q ≤ Q^{n-δ}$  and\label{item:1} 
  \item $\dist_ℤ(q β_i) ≤ {1}/{Q^{1+δ}}$ for each $i ∈ \{1,\dots,n\}$. \label{item:2}
  \end{enumerate}
  The condition \ref{item:1}) implies that
  \begin{displaymath}
    q^{1+δ} ≤  Q^{(n-δ)(1+δ)} < Q^{n(1+δ)}. 
  \end{displaymath}
  Together with \ref{item:2}) this implies that
  \begin{equation} 
    q^{1+δ} \dist_ℤ(q β_1) \cdots \dist_ℤ(q β_n) <1.  \label{eq:12}
  \end{equation}

  By theorem~\ref{thr:3}, the number of integral $q$ verifying \eqref{eq:12} is finite.
	Furthermore, given that $\beta_1$ is a square root of an integer, we have that
		\[ \frac{1}{q \cdot (2\beta_1+1)} \leq \dist_\Z(q \beta_1) \leq  \frac{1}{Q^{1+δ}}. \]
	Therefore, $q$ is bounded from below by an increasing function of $Q$. As there are finitely many $q$ verifying \eqref{eq:12}, there are also finitely many $Q$ for which $λ_1 ≤ \frac{1}{Q^{1+δ}}$.
 Therefore, there exists a bound $Q_0$ such that for all $Q ≥ Q_0$ the successive minimum  $λ_1$ of $Λ(B)$ satisfies $λ_1 ≥ 1/Q^{1 + \frac{1}{3n}}$. 
  
  We now also consider the case $a_1$ is equal to $1$.  In this case, $β_1 = 1$ and the lattice basis $B$ is given by
  \begin{displaymath} B = 
     \begin{pmatrix}
     1/Q^{n+1} &  &  & &   \\
     1      &  1            &        &   & \\
     β_2    &   & 1  &&\\
     \vdots     &       &         & \ddots & \\
     β_n    &           &     &        & 1
   \end{pmatrix}
  \end{displaymath}
It is easy to see that the first successive minimum of $Λ(B)$ remains the same upon deletion of the second row and column. Denote this updated basis by $B' ∈ ℝ^{n ×n}$. With exactly the same argument as above, it follows that there exists a $Q_0$ such that, for all $Q ≥ Q_0$ one has $λ_1$ of $Λ(B')$ satisfies $λ_1 ≥ 1/Q^{1 + \frac{1}{3n}}$.

Finally, by choosing $Q = (Q_0 \cdot  \left(2n \|x\|_\infty\right)^{3/2})$, we obtain the lower bound dependent on $Q_0$ and a single exponential in $n \cdot \|x\|_\infty$. Moreover, $Q_0$ depends only on the bound for the finitely many $q$ verifying the statement of Theorem \ref{thr:3}. This means that $Q_0$ is a constant depending only on $a_1,\dots,a_n$.  
  
\end{proof}

\begin{remark}
	The exponent $2n$ of equation \eqref{eq:11} can be decreased to any $n+\epsilon$, with $\epsilon>0$ using a suitable $\delta$ when applying Theorem \ref{thr:3}.
	Note that this would affect the constant $\gamma$ by making it dependent on $\epsilon$.
\end{remark}

\section{An upper bound via number balancing}
\label{sec:lower-bound-via}

We now show asymptotic (almost) tightness of the bound on $|E|$ when $a_1,\dots,a_n$ are fixed and distinct square-free positive integers. Since $\sqrt{a_1}, \cdots, \sqrt{a_n}$ are linearly independent over $ℚ$, $E$ is nonzero whenever $x ∈ ℤ^n$  is not equal to zero. We show that there exist solutions asymptotically (almost) tight in $\| x\|_\infty$ via the pigeon-hole principle, as it is used in the  \emph{number balancing problem}~\cite{karmarkar1982differencing,hoberg2017number}.  

\begin{theorem}
  \label{thr:5}
  Let $L \geq 2$. There exists a nonzero $x ∈ ℤ^n$  with  $ \|x\|_∞≤ L$  such that
  \begin{displaymath}
    |E| ≤ \frac{n \max_i \sqrt{a_i}}{ L^{n-1}}. 
  \end{displaymath}
\end{theorem}

\begin{proof} 
  Let $\beta = (\sqrt{a_1}, \dots, \sqrt{a_n}) \in \R^n$.
The number of vectors $y \in \Z^n$ such that $\|y\|_∞ ≤ L/2$ holds is
  at most $L^n$. On the other hand one always has \[ |y^T \beta | ≤ \frac{n L}2 \max_i \sqrt{a_i} \] for such vectors $y$.
 By the pigeon-hole principle, there exist $y_1\neq y_2 ∈\Z^n$  of infinity norms at most $L/2$ such that their corresponding values are close:
	\[ | y_1^T \beta - y_2^T \beta | \leq \frac{2 n L}{2(L^n - 1)} \max_i \sqrt{a_i}. \]
The difference $x = y_1 - y_2$ hence verifies $\|x \|_\infty \leq L$ and the required bound.
\end{proof}

\section{Discussion}
\label{sec:discussion}

  The subspace theorem (Theorem~\ref{thr:3}) does not provide explicit bounds on the number of solutions $q ∈ ℕ_+$. The existing \emph{quantitative} versions of the subspace theorem, see, e.g.~\cite{evertse2013further}, do not provide such bounds either.  This is still the case when all algebraic numbers are square roots of integers. An explicit bound on the number of solutions would immediately apply to a separation bound for the sum of square roots.
  
  In light of the relationship of the subspace theorem and separation bounds that we describe in this paper, it is an interesting open problem to find explicit upper and lower bounds on the number of solutions $q ∈ ℕ_+$ satisfying the equations of Theorem~\ref{thr:3} for $β_i = \sqrt{a_i}$ and $δ = 1/ \poly(n)$.  

\bibliographystyle{plain}
\bibliography{mybib}

\end{document}